\documentclass[a4paper, 11pt]{amsart}
\newtheorem{lemma}{Lemma}

\theoremstyle{remark}

\newtheorem{definition}{Definition}
\newtheorem{claim}{Claim}

\newtheorem{example}{Example}

\usepackage{amsmath, amssymb, amsthm}
\usepackage{tikz}
\usepackage{enumerate}
\usepackage{hyperref}
\usepackage[norelsize,algoruled,vlined,algo2e]{algorithm2e}
\usetikzlibrary{arrows}

\newcommand{\lroots}{\Lambda}
\newcommand{\B}{B}

\begin{document}
\title{Beyond the Runs Theorem}
\author{
Johannes Fischer}
\address[J. Fischer, Tomohiro I]{Department of Computer Science, TU Dortmund, Germany}
\email{johannes.fischer@cs.tu-dortmund.de}
\email{tomohiro.i@cs.tu-dortmund.de}

\author{\v St\v ep\' an Holub}
\address[\v{S}. Holub]{Department of Algebra, Charles University, Sokolovsk\'a 83, 175 86 Praha, Czech Republic}
\email{holub@karlin.mff.cuni.cz}

\author{Tomohiro I}

\author{Moshe Lewenstein}

\address[M. Lewenstein]{Department of Computer Science, Bar-Ilan University}
\email{moshe@cs.biu.ac.il}

\thanks{J. Fisher and M. Lewenstein supported by a Grant from the GIF, the German-Israeli Foundation for Scientific Research and Development}
\thanks{\v{S}. Holub supported by the Czech Science Foundation grant number 13-01832S}

\begin{abstract}
In \cite{runstheorem}, a short and elegant proof was presented showing that a binary word of length $n$ contains at most $n-3$ runs. Here we show, using the same technique and a computer search, that the number of runs in a binary word of length $n$ is at most $\frac{22}{23}n<0.957n$. 
\end{abstract}
\maketitle

\section{Introduction}
The research on the possible (maximal) number of runs in a word of length $n$ dates back at least to \cite{KK}. Since then, there where two types of efforts: finding words rich of runs, and proving an upper bound on the number of runs. Both efforts were accompanied by a heavy use of computer search.  An (at least psychologically) important barrier was the question whether the number of runs can be larger than the length of the word, and the negative answer was known as ``the runs conjecture''. The barrier was broken, turning the conjecture into a theorem, by a  remarkably simple and computer-free proof in \cite{runstheorem}.
In this paper we continue the narrowing of the gap between the two bounds. We build essentially on the technique leading to the beautiful proof of the Runs Theorem, adding again some computer backing.

For the more detailed description of the history of the problem and for an extensive list of literature, see for example \cite{Crochemore,runstheorem}.

\section{Runs and Lyndon roots}
For any word $u$, an integer $p$ with $1 \le p \le |u|$ is said to be a \emph{period} of $u$ if $u[i] = u[i+p]$ for all $1 \leq i \leq |u|-p$.
Especially, the smallest period of $u$ is called \emph{the period} of $u$.
A prefix $v$ of $u$ that is also a suffix of $u$ is said to be a \emph{border} of $u$.
The empty word and $u$ are trivial borders of $u$.
We call $u$ \emph{unbordered} if there is no border other than trivial ones.

Given a word $w$, we say that an interval $[i..j]$ with $1\leq i\leq j\leq |w|$ is \emph{period-maximal in $w$} if $w[i..j]$ has no extension in $w$ with the same period. That is, if $1\leq i'\leq i\leq j\leq j'\leq |w|$ is such that $w[i..j]$ and $w[i'..j']$ have the same period, then $i=i'$ and $j=j'$.
A period-maximal interval is said to be \emph{left-open} if $i = 1$, otherwise it is \emph{left-closed}. Similarly, a period-maximal interval is \emph{right-open} or \emph{left-closed} depending on whether or not $j = |w|$.
If $1 < i$ and $j < |w|$, the interval is said to be \emph{closed}.
A period-maximal interval is a \emph{run} if its length is at least double of the period $p$ of $w[i..j]$, that is $j-i+1\geq 2p$.

We shall work with the two-letter alphabet $\{0,1\}$, which allows two lexicographic orders: $\prec_0$ is defined by $0\prec_0 1$, and $\prec_1$ by $1\prec_1 0$.
We shall write $\bar a=1-a$.
A word $v$ is said to be a \emph{Lyndon word} with respect to some order $\prec$ if and only if $w\prec u$ for any nonempty proper suffix $u$ of $w$. In particular, Lyndon words are unbordered. We say that a Lyndon word $v$ is a \emph{Lyndon root} of $w$ if $v$ is a factor of $w$ and $|v|$ is the period of $w$.

A right-closed period-maximal interval $[i..j]$ of $w$ is said to be \emph{$a$-broken} in $w$, where $a=w[j+1]$. We will also say, a bit imprecisely, that the period of $w[i..j]$ is broken by $a$.

Let $\rho(n,2)$ denote the maximal number of runs in a binary word of length $n$.
\medskip

The basic idea of~\cite{runstheorem} is to associate an $a$-broken run $r=[i..j]$ with the set $\lroots(r)$ of intervals corresponding to the Lyndon root of $w$ with respect to the order $a\prec \bar a$, excluding from $\lroots(r)$, if necessary, the interval starting at the beginning of $r$. This definition has to be completed to cover also runs that are not broken, that is, right-open runs. For those runs, the set $\lroots(r)$ can be defined as consisting of Lyndon roots with respect to both orders. In \cite{runstheorem}, the case of unbroken runs is solved by appending a special symbol $\$$ to the end of $w$, which is equivalent to arbitrarily choosing one of the orders (the order $0\prec 1$ in their case).

Let ${\tt Beg}(S)$ denote the set of starting positions of intervals in the set $S$, and let $\B(r) = {\tt Beg}(\lroots(r))$ for any run $r$.
The crucial fact, implying instantaneously that there are at most $|w|-1$ runs, is that $\B(r)$ and $\B(r')$ are disjoint for $r\neq r'$.
At no cost, it is possible to make this basic tool a bit stronger. 
For sake of clarity, let us first give a formal definition.

\begin{definition}\label{def:lroots}
Let $w$ be a binary word. Let $s=[i..j]$ be a period-maximal interval in $w$ with period $p$. Then $\lroots(s)$ denotes the set of all intervals $[i'..j']$ of length $p$ such that $i< i'\leq j'\leq j$ and $w[i'..j']$ is a Lyndon word with respect to an order $\prec$ satisfying the following condition: if $j<|w|$ and $[i..j]$ is $a$-broken in $w$, then $a\prec \bar a$ (the condition being empty if $[i..j]$ is not broken). Also, let $\B(s) = {\tt Beg}(\lroots(s))$.
\end{definition}

\begin{example}
Take a word $w = 1110101101$ of length $10$.
For a period-maximal interval $s_1=[1..3]$ with period $1$, $\lroots(s_1) = \{ [2..2], [3..3] \}$.
For a period-maximal interval $s_2=[3..7]$ with period $2$, $\lroots(s_2) = \{ [5..6] \}$.
For a period-maximal interval $s_3=[5..10]$ with period $3$, $\lroots(s_3) = \{ [6..8], [7..9] \}$.
Note that $w[6..8] = 011$ and $w[7..9] = 110$ are Lyndon words w.r.t. $\prec_{0}$ and $\prec_{1}$, respectively.
For a period-maximal interval $s_4=[2..10]$ with period $5$, $\lroots(s_4) = \{ [4..8] \}$.
\end{example}

The following lemma is now stronger than the corresponding \cite[Lemma~8]{runstheorem} in two ways. First, it applies also to period-maximal intervals that are not runs, and second, as noted above, $\lroots(r)$ is defined more generously for unbroken runs. The proof, however, is the same.

\begin{lemma}\label{lem1}
Let $s$ and $t$ be two distinct period-maximal intervals in $w$. Then $\B(s)$ and $\B(t)$ are disjoint.
\end{lemma}
\begin{proof}
Let $s=[i_s..j_s]$ and $t=[i_t..j_t]$.
Suppose that $k\in \B(s)\cap \B(t)$, and let $[k..m_s]\in \lroots(s)$ and $[k..m_t]\in \lroots(t)$.
If $m_s=m_t$, then $s$ and $t$ have the same period and $s=t$. We can therefore, w.l.o.g., suppose that $m_s< m_t$. Then $w[k..j_s]$ has a smaller period than the unbordered $w[k..m_t]$, which implies that $j_s<m_t$. Therefore $s$ is $a$-broken with $a=w[j_s+1]$.

Since $a$ breaks the period of $w[k..j_s]$, 
we have $w[m_s+1..j_s+1] \prec_a w[k..j_s+1]$. 
Since both $w[m_s+1..j_s+1]$ and $w[k..j_s+1]$ are factors of $w[k..j_t]$, we deduce that $w[k..m_t]$ is Lyndon w.r.t. $\prec_{\bar a}$. Note that $w[k..j_t]$ contains both letters. Therefore $w[k]=\bar a$, and the $\prec_a$-minimality of $w[k..m_s]$ implies that $w[i_s..j_s]\in {\bar a}^+$. The definition of $\lroots(s)$ yields $i_s<k$ and $i_t<k$, which leads to a contradiction with $\prec_{\bar a}$-minimality of $w[k..m_t]$. 
\end{proof}

\begin{example}
It is worth noting that the appearance of $\bar a^+$ in the previous proof is significant, and it is the place where we use the prohibition of the very first position of a run. Without this condition, Lemma~\ref{lem1} would not hold. Consider the word $1101011$ and position $2$, which is the starting point of the Lyndon root $1$ of the run $11$ and the starting point of the Lyndon root $10$ of $10101$, the latter being excluded by the prohibition. 
\end{example}

Lemma~\ref{lem1} implies that for each position $k$ there is at most one period-maximal interval $s$ such that $k\in \B(s)$.
% Such an $s$ can be found using the following rules (compare with \cite[Lemma~6]{runstheorem}).
% For any $k > 1$ (with possibly one exception), such an $s$ can be found using the following rules.
Such an $s$ can be found using the following rules.

\begin{lemma}\label{llyndon}
Let $k>1$ be a position of $w$ such that $w[k] = a$ and $w[k-1..|w|] \neq \bar a a^+$. Then $k\in \B(s)$ where $s=[i..j]$ is the period-maximal extension of 
\begin{itemize}
  \item $[k..k]$, if $w[k]=w[k-1]$;
  \item $[k..k']$, where $w[k..k']$ is the longest Lyndon word with respect to $\prec_a$ starting at the position $k$, otherwise.
\end{itemize}
\end{lemma}
\begin{proof}
If $w[k]=w[k-1]$, then $s$ is a run with period one containing the position $k$ with $i < k$.
Hence, $k\in \B(s)$ immediately follows from Definition~\ref{def:lroots}.

Let $w[k-1] = \bar a$, and let $w[k..k']$ be the longest Lyndon word with respect to $\prec_a$ starting at the position $k$.  From $w[k..|w|] \neq a^+$, it is easy to see that $k'\neq k$ and $w[k']=\bar a$, which implies $i < k$.
If $s$ is right-open, we are done:
$k \in \B(s)$ since the condition on $\prec$ is empty (see Definition~\ref{def:lroots}).
It remains to show that $s$ is $a$-broken if it is broken.
Assume to the contrary that $s$ is $\bar a$-broken. We show that $w[k..j+1]$ is a Lyndon word with respect to $\prec_a$.
Let $p$ denote the length of the Lyndon word $w[k..k']$, that is, $p=k'-k+1$. 
Let first $k < h \le k'$. 
% for some positive integer $q$. 
Since $w[k..k']$ is a Lyndon word with respect to $\prec_a$, we have $w[k..k']\prec_a w[h..k']$, and thus also $w[k..j+1] \prec_a w[h..j+1]$.
Let now $k'<h\leq j+1$. As above, $w[k..k']\prec w[h-p+1..h] \prec_a w[h-p+1..j+1]$. Also $w[h-p+1..j+1] \prec_a w[h..j+1]$, since $w[h..j+1]=w[h..j]\bar a$ and $w[h..j]a$ is a prefix of $w[h-p+1..j+1]$. 
Therefore, $w[k..j+1]$ is a Lyndon word, which contradicts that $w[k..k']$ is the longest Lyndon word starting at the position $k$.
\end{proof}

Note that for the position $k$ with $w[k-1..|w|] = \bar a a^+$, there is no period-maximal interval $s$ with $k \in \B(s)$.
An algorithm computing for all positions the longest Lyndon words starting there is discussed in~\cite[Section~4.1]{runstheorem}.

\section{Idle positions}
In order to make explicit the relation between runs and positions, we associate with a run $r$ the position $\max \B(r)$ and say that such a position is \emph{charged} (by $r$). 
We repeat that the Runs Theorem was proved in~\cite{runstheorem} by pointing out that charging is an injective mapping, which is a corollary of Lemma~\ref{lem1}.
This also yields an obvious strategy for further lowering the upper bound on the number of runs. One has to find positions that are not charged in an arbitrary word. We shall call such positions \emph{idle}.
Equivalently, we want to identify a position $i$ satisfying either of the following two conditions.
\begin{enumerate}
	\item $i$ is not contained in $\B(r)$ for any run, or
	\item $i$ is in $\B(r) \setminus \{\max \B(r)\}$ for some run $r$. %(we shall call runs with $|\B(r)|>1$ \emph{overloaded}).
\end{enumerate}

\subsection{Idle positions that are resistant to extensions}
In order to be able to estimate the number of idle positions locally, we are interested in idle positions that remain idle in any extension of $w$. 
One obvious fact is that closed period-maximal intervals are not affected by extensions.
For example, the third position in the word $1010011$ remains idle for any extensions.
That is because the period three of $1001$ is broken by $1$, and the period-maximal extension of $1001$ is $s=[2..6]$ that is closed,
but $s$ is not a run, and Definition~\ref{def:lroots} and Lemma~\ref{lem1} yield that the position is idle.

Also, it is easy to see that runs $r$ with $|\B(r)|>1$ that are right-closed preserve this property in any extension.
However, we have to be careful with right-open runs since some positions in $\B(r)$ may disappear when the run $r$ gets broken by a right-extension.
To clarify this case, let $\lroots_a(r)$ denote the set of Lyndon roots in $\lroots(r)$ that are Lyndon words with respect to $\prec_a$,
and let $\B_a(r) = {\tt Beg}(\lroots_a(r))$.
Note that $\B_a(r) = \B_{\bar a}(r)$ if and only if $r$ is a run with period one.
Now we consider the set $D(w)$ of idle positions $k$ in a word $w$ falling into one of the following cases:
\begin{enumerate}[(a)]
  \item $k \in B(s)$, where $s$ is a closed period-maximal interval that is not a run.
  \item $k \in (B_a(r) \setminus \{ \max B_a(r) \})$, where $r$ is an $a$-broken run.
  \item $k \in (B_a(r) \setminus \{ \max B_a(r) \})$, where $r$ is a right-open run and $a$ is chosen such that
        $\min B_a(r) \ge \min B_{\bar a}(r)$ ($a \in \{0, 1\}$ is arbitrary if its period is $1$).
\end{enumerate}
By $D(w)$ we intend to say that, for any $k \in D(w)$, the position $|u|+k$ in $uwv$ is idle for any extensions $u$ and $v$.
The only exception is the case (c) in which the position $|u|+k$ may not be idle if $r$ is $\bar a$-broken in the extension.
But even in this case we have that at least one of the positions $|u|+k$ and $|u|+k-g$ of $uwv$ is idle,
where $g = \min B_{a}(r) - \min B_{\bar a}(r)$.
Therefore the number of idle positions does not decrease for any extensions.
This is formulated in the following claim.
\begin{claim} \label{D}
Let $w$, $u$ and $v$ be arbitrary binary words. Then
\begin{align*}
	\Big| D(uwv)\cap [|u|+2..|uw|-1] \Big| \geq  |D(w)|.
\end{align*}
\end{claim}
\begin{proof}
  We examine $k \in D(w)$ of each case:
  \begin{itemize}
    \item For Case (a). Since $s=[i..j]$ is closed,
          we have a closed period-maximal interval $s' = [|u|+i..|u|+j]$ in $uwv$.
          Since $s'$ is not a run, $|u|+k$ is in $D(uwv)$.
    \item For Case (b). Since $r=[i..j]$ is an $a$-broken run, 
          we have an $a$-broken run $r' = [i'..|u|+j]$ with $i' \le |u|+i$ in $uwv$.
          Since any Lyndon root in $\lroots(r)$ appears in $\lroots(r')$ (with shift $|u|$), $|u| + k$ is in $D(uwv)$.
    \item For Case (c). Let $r = [i..|w|]$ be a right-open run in $w$. We have a run $r' = [i'..j']$ with $i' \le |u|+i$ and $|uw| \le j'$ in $uwv$.
          Note that $k-g \in \B_{\bar a}(r)$, where $g = \min B_{a}(r) - \min B_{\bar a}(r)$.
    \begin{itemize}
      \item If $r'$ is still open or $a$-broken, any Lyndon root in $\lroots_{a}(r)$ appears in $\lroots_{a}(r')$ (with shift $|u|$),
            and hence, $|u|+k$ is in $D(uwv)$.
      \item If $r'$ is $\bar a$-broken, any Lyndon root in $\lroots_{\bar a}(r)$ appears in $\lroots_{\bar a}(r')$ (with shift $|u|$),
            and hence, $|u|+k-g$ is in $D(uwv)$.
    \end{itemize}
  \end{itemize}
We have described an injective map from $D(w)$ to $D(uwv)\cap [|u|+1..|uw|]$. The map always assigns, for some $a$ and some $r$, a position $k$ in $[1..|w|]\cap B_a(r)$ to the position $k+|u|$. Taking into account that $1$ and $|w|$ are not in any $B_a(r)$, we get the claim.
\end{proof}

This yields the following lemma:
\begin{lemma}\label{lem:lim_ub}
  If $|D(w)| \ge d$ for any binary word $w$ of length $m$, then
  \[\lim_{n \rightarrow \infty} \left(\frac {\rho(n, 2)}{n}\right) \le \frac{m-2-d}{m-2}.\]
\end{lemma}
\begin{proof}
  Let $y=ay_1y_2\cdots$ be an infinite binary word, where $a$ is a letter, and $|y_i|=m-2$ for each $i$. By Claim~\ref{D},
  each interval corresponding to a factor $y_i$ in $y$ contains at least $d$ idle positions. The claim follows.
\end{proof}

\subsection{Idle positions that are resistant to left extensions}
We further identify positions that remain idle when we consider ``only'' left extensions,
which only comes into play in Section~\ref{sec:omit_limit} to estimate the number of idle positions in a suffix of a word.
Formally, for any word $w$ we define the set $D'(w)$ of idle positions $k$ in $w$ falling into one of the following cases:
\begin{enumerate}[(A)]
  \item $k \in \max B(s)$, where $s$ is a left-closed period-maximal interval that is not a run.
  \item $k \in (B(s) \setminus \{ \max B(s) \})$, where $s$ is a period-maximal interval (which is possibly a run).
  \item $w[k-1..|w|] = \bar a a^+$ holds.
\end{enumerate}

Note that $D(w) \subseteq D'(w)$.
Since we do not consider right-extensions,
we can show the following claim, which is a bit stronger than Claim~\ref{D} for $D$.
\begin{claim} \label{D'pos}
Let $w$ and $u$ be arbitrary binary words. For any $k \in D'(w)$, $|u|+k \in D'(uw)$.
\end{claim}
\begin{proof}
  We examine $k \in D'(w)$ of each case:
  \begin{itemize}
    \item For Case (A). Since $s=[i..j]$ is left-closed,
          we have a left-closed period-maximal interval $s' = [|u|+i..|u|+j]$ in $uw$.
          Since $s'$ is not a run, $|u|+k$ is in $D'(uw)$.
    \item For Case (B). Let $s = [i..j]$.
          We have a period-maximal interval $s' = [i'..|u|+j]$ with $i' \le |u|+i$ in $uw$.
          Since any Lyndon root in $\lroots(s)$ appears in $\lroots(s')$ (with shift $|u|$), $|u|+k$ is in $D'(uw)$.
    \item For Case (C). Since $\bar a a^+$ stays in a suffix of $uw$, $|u|+k$ in $D'(uw)$.
  \end{itemize}
\end{proof}

Considering that $1 \notin D'(w)$, we get:
\begin{claim} \label{D'num}
Let $w$ and $u$ be arbitrary binary words. Then
\begin{align*}
	\Big| D'(uw) \cap [|u|+2..|uw|] \Big| \geq  |D'(w)|.
\end{align*}
\end{claim}

\section{Computer search}
Given a positive integer $d$, Algorithm~\ref{algo:compute_m} computes the minimum integer $m_d$ such that 
$|D(w)| \ge d$ for any binary word $w$ of length $m_d$. 
The algorithm traverses words by appending characters to the right.
If $|D(w)| \ge d$, we stop the extension
since $|D(wv)| \ge d$ for any word $v$.

If we already know the value $m_{d'}$ for some $d' < d$, then the following pruning of the search space can be employed:
If $|D(w) \cap [1..m - m_{d'} + 1]| \ge d - d'$, then we stop the extension.
That is because for any word $z$ of length $m_{d'}$, $D(z)$ contains at least $d'$ positions (and $1\notin D(z)$), and hence,
for any word $v$ of length $m - |w|$, $D(wv)$ contains at least $d - d'$ positions in $[1..m - m_{d'} + 1]$ and at least $d'$ positions in $[m - m_{d'} + 2..m]$.
Namely, $D(wv)$ contains at least $d-d'+d'=d$ positions, and hence, any right extension of the current $w$ cannot lead to an update of $m$.

\begin{algorithm2e}[t]
  \caption{Computing $m_d$.}
  \label{algo:compute_m}
    \KwIn{A positive integer $d$.}
    \KwOut{Return $m_d$.}
    \SetKwFunction{Extend}{Extend}
    \SetKw{Procedure}{procedure}
\nl $m \leftarrow 0$; \tcp{Let $m$ be a global variable.}
\nl \Extend{$0$}\;
\nl \Return $m$\;

\BlankLine

\Procedure \Extend{$w$}\;
\nlset{1} \lIf{$|w| > m$}{$m \leftarrow |w|;$}

\nlset{2} compute $D(w)$\;
\nlset{3} \lIf{$|D(w)| \ge d$}{\Return;}\label{algoline:pruning}

\nlset{4} \ForEach{$a \in \{0, 1\}$}{
\nlset{5}   \Extend{$wa$}\;}
\end{algorithm2e}

By computing $m_d$ and using Lemma~\ref{lem:lim_ub}, we obtained upper bounds for $\lim_{n \rightarrow \infty}(\rho(n, 2)/n)$ given in Table~\ref{table:ub}.

\begin{table}[t]
  \caption[]{Upper bounds of $\lim_{n \rightarrow \infty}(\rho(n, 2)/n)$.}
  \label{table:ub}
   \centering{
    \begin{tabular}{|r|r|l|r|}
      $d$ & $m_d$ & $\lim_{n \rightarrow \infty}(\rho(n, 2) / n)$ & $m_{d} - m_{d-1}$ \\ \hline
      1   & 63    & 0.98360655737\ldots & 63 \\
      2   & 96    & 0.97872340425\ldots & 33 \\
      3   & 126   & 0.97580645161\ldots & 30 \\
      4   & 150   & 0.97297297297\ldots & 24 \\
      5   & 172   & 0.97058823529\ldots & 22 \\
      6   & 194   & 0.96875             & 22 \\
      7   & 216   & 0.96728971962\ldots & 22 \\
      8   & 237   & 0.96595744680\ldots & 21 \\
      9   & 258   & 0.96484375          & 21 \\
      10  & 274   & 0.96323529411\ldots & 16 \\
      11  & 295   & 0.96245733788\ldots & 21 \\
      12  & 314   & 0.96153846153\ldots & 19 \\
      13  & 332   & 0.96060606060\ldots & 18 \\
      14  & 351   & 0.95988538681\ldots & 19 \\
      15  & 369   & 0.95912806539\ldots & 18 \\
      16  & 388   & 0.95854922279\ldots & 19 \\
      17  & 407   & 0.95802469135\ldots & 19 \\
      18  & 425   & 0.95744680851\ldots & 18 \\
      19  & 444   & 0.95701357466\ldots & 19 \\
      20  & 462   & 0.95652173913\ldots & 18 \\\hline
    \end{tabular}
   }
\end{table}

\section{Upper bound for finite words}\label{sec:omit_limit}
We now prove that we can omit the limit in the bounds in Table~\ref{table:ub}.
That is, we verify that, for any $d \le 20$, $\rho(n, 2)/n \le (m_d-2-d)/(m_d-2)$ does hold for any $n$.

Let $y$ be a finite word and let $p_1$, $p_2$, \dots, $p_\ell$ be the list of idle positions of $y$. Note that $p_1=1$.
For a given $d$ we define 
\begin{align*}
s_k&=s_k(y,d):=[p_{(k-1)d+1}..p_{kd+1}-1] & &\text{for $k=1,2,\dots, \lceil\ell/d \rceil-1$},\\
s_k&=s_k(y,d):=[p_{(k-1)d+1}..|y|] & &\text{for $k=\lceil\ell/d \rceil$}\,.
\end{align*}
In other words, we make a disjoint decomposition of the interval $[1..|y|]$ into subintervals $s_k$ such that each $s_k$ starts with an idle position of $y$, and each $s_k$, except maybe the last one, contains exactly $d$ idle positions.

We first claim that all intervals $s_k$, $k<\lceil\ell/d \rceil$, have length at most $m_d-2$. Suppose that the length of some $s_k=[i..j]$ is at least $m_d-1$ and consider the word $y[i..j+1]$ of length $m_d$. By the definition of $m_d$ and by Claim~\ref{D}, the cardinality of $D(y)\cap [i+1..j]$ is at least $d$ which means that $[i..j]$ contains at least $d+1$ idle positions, a contradiction.

It remains to count idle positions in the tail of the word $y$, that is, in the interval $s_{\lceil\ell/d \rceil}$. By an argument similar to the one above, one can see that the length of the interval is at most $m_d-1$.
Let $z$ denote the suffix in question, that is, $z=y[p_{(\lceil\ell/d \rceil-1)d+1}..|y|]$.
Since we only have to consider left-extensions of $z$, we now use $D'(z)$ to estimate the number of idle positions.
Since $1\notin D'(z)$ and the first position of $s_{\lceil\ell/d \rceil}$ is idle in $y$, our goal is to show  
\begin{align}
	\frac{|z|-|D'(z)|-1}{|z|} < \frac{m_d-2-d}{m_d-2}\,. \tag{$*$}\label{dif}
\end{align}

Let $d=20$. Then the right hand of \eqref{dif} is $22/23$. We first note that $(x-1)/x<22/23$ for each $x< 23$. Therefore, we can assume $|z|\geq 23$.

A simple computer search verified that $|D'(w)| \ge 3$ for any word $w$ with $|w| \ge 13$,
which means there are at least $3$ idle positions in the last $12$ positions of $w$ that are resistant to left extensions.
Let now $z=z_1z_2$ with $|z_2| = 12$.
If $m_i-1 \leq |z_1| < m_{i+1}-1$ (where $m_0:=0$), then $z$ has at least $i$ idle positions in $[2..|z_1|]$ by Claim~\ref{D},
and hence, $|D'(z)| \ge i + 3$.
Using the results in Table~\ref{table:ub},
a direct calculation yields that, for each $i=0,1,\dots,19$, if $m_i-1 \leq |z_1| < m_{i+1}-1$, then
\begin{align*}
	\frac{|z|-|D'(z)|-1}{|z|}
	&\leq \frac{(m_{i+1}-2+12)-i-3-1}{m_{i+1}-2+12} <\frac{22}{23}\,.
\end{align*}

Therefore we get the following result.
\begin{align*}
 \rho(n, 2)/n & < \frac{22}{23} =
 0.\overline{9565217391304347826086}\,.
\end{align*}

 \section{Conclusion}
 Search for words with high number of runs in the literature yields words with approximately $0.944n$ runs, where $n=|w|$, see \cite{jamie, lower, web}. Therefore, the optimal multiplicative constant is somewhere between $0.944$ and $0.957$. The lower bound corresponds to words where on average about every 18th position is idle. This seems to fit very well with the eventual distances between $m_{d-1}$ and $m_{d}$ in Table~\ref{table:ub}. It is therefore reasonable to expect that the optimal density of runs is close to the lower bound, maybe around $1-1/18.5\approx 0.946$.

\end{document}